\documentclass[letterpaper,11pt]{article}
\usepackage[utf8]{inputenc}

\pdfoutput=1

\usepackage{amsmath, amsthm, amssymb}
\usepackage{algpseudocode,algorithm}
\usepackage{algorithmicx}
\usepackage{mathtools}
\usepackage[numbers]{natbib}
\usepackage{comment} % enables comment environment
\usepackage{color-edits} % author comments with colors
\usepackage[most]{tcolorbox}
\usepackage{xfrac}
\usepackage{hyperref}
\usepackage{multirow}
\usepackage{caption}
\usepackage{bm}
\usepackage{newfloat}
\usepackage{enumitem}
\usepackage{xpatch}

\makeatletter
\xpatchcmd{\algorithmic}
  {\ALG@tlm\z@}{\leftmargin\z@\ALG@tlm\z@}
  {}{}
\makeatother

\usepackage[margin=1in]{geometry}

\allowdisplaybreaks

\definecolor{mygreen}{RGB}{20,120,60}

\hypersetup{
     colorlinks=true,
     citecolor= mygreen,
     linkcolor= mygreen
}

\title{\vspace{-0.5cm}Massively Parallel Symmetry Breaking on Sparse Graphs:\\ MIS and Maximal Matching\footnote{A merger of this paper and the concurrent paper of Brandt, Fischer, and Uitto \cite{concurrent} appeared at PODC 2019 \cite{merger}.}}

\date{}

\author{Soheil Behnezhad\thanks{University of Maryland. Emails: \texttt{\{soheil, mahsa, hajiagha\}@cs.umd.edu}.} \and Mahsa Derakhshan\footnotemark[2] \and MohammadTaghi Hajiaghayi\footnotemark[2] \and Richard M. Karp\thanks{UC Berkeley. Email: \texttt{karp@cs.berkeley.edu}.}}

\newcommand{\MPC}[0]{\ensuremath{\mathsf{MPC}}}
\newcommand{\local}[0]{\ensuremath{\mathsf{LOCAL}}}
\newcommand{\congest}[0]{\ensuremath{\mathsf{CONGEST}}}
\newcommand{\PRAM}[0]{\ensuremath{\mathsf{PRAM}}}
\DeclareMathOperator{\poly}{poly}

\newcommand{\Ot}[1]{\ensuremath{\widetilde{O}(#1)}}

\renewcommand{\O}[1]{\ensuremath{O(#1)}}

\addauthor{sb}{blue}    % sb for Soheil
\addauthor{md}{red}  % md for Mahsa
\addauthor{ss}{purple} %ss for Saeed

\newtheorem{theorem}{Theorem}
\newtheorem{lemma}{Lemma}[section]

\newtheorem{definition}[lemma]{Definition}

\newtheorem{observation}[lemma]{Observation}
\newtheorem{remark}[lemma]{Remark}

\definecolor{mygreen}{RGB}{20,140,80}
\definecolor{myred}{RGB}{110,0,0}
\definecolor{mylightgray}{RGB}{230,230,230}

\hypersetup{
     colorlinks=true,
     citecolor= mygreen,
     linkcolor= myred
}

\newcommand{\smparagraph}[1]{\vspace{0.2cm}\noindent\textbf{#1}}

\algnewcommand{\IIf}[1]{\State\algorithmicif\ #1\ \algorithmicthen}
\algnewcommand{\EndIIf}{\unskip\ \algorithmicend\ \algorithmicif}

\newcounter{myalgctr}
\newenvironment{tbox}{
\vspace{0.2cm}
\begin{tcolorbox}[width=\textwidth,
                  enhanced,
                  boxsep=2pt,
                  left=1pt,
                  right=1pt,
                  top=4pt,
                  boxrule=1pt,
                  arc=0pt,
                  colback=white,
                  colframe=black,
                  breakable
                  ]%%
}{
\end{tcolorbox}
}

\newenvironment{graytbox}{
\vspace{0.2cm}
\begin{tcolorbox}[width=\textwidth,
                  enhanced,
                  frame hidden,
                  boxsep=6pt,
                  left=1pt,
                  right=1pt,
                  top=4pt,
                  boxrule=1pt,
                  arc=0pt,
                  colback=mylightgray,
                  colframe=black,
                  breakable
                  ]%%
}{
\end{tcolorbox}
}

\newcommand{\tboxhrule}[0]{\vspace{0.1cm} \hrule \vspace{0.2cm}}

\newenvironment{titledtbox}[1]{\begin{tbox}#1 \tboxhrule}{\end{tbox}}

\newenvironment{tboxalg}[1]{\refstepcounter{myalgctr}\begin{titledtbox}{\textbf{Algorithm \themyalgctr.} #1}}{\end{titledtbox}}

\newenvironment{quotebox}{
\begin{tcolorbox}[width=\textwidth,
                  enhanced,
                  frame hidden,
                  interior hidden,
                  boxsep=0pt,
                  left=35pt,
                  right=35pt,
                  top=4pt,
                  boxrule=1pt,
                  arc=0pt,
                  colback=white,
                  colframe=black
                  ]%%
\itshape
}{
\end{tcolorbox}
}

\newcommand{\restate}[2]{
\vspace{0.3cm}
\noindent \textbf{#1.} (restated) {\em #2}
\vspace{0.2cm}
}

\begin{document}
\maketitle

\begin{abstract}
\setlength{\parskip}{0.4em}
	The success of modern parallel paradigms such as MapReduce, Hadoop, or Spark, has attracted a significant attention to the {\em Massively Parallel Computation} (\MPC{}) model over the past few years, especially on graph problems. In this work, we consider  symmetry breaking problems of {\em maximal independent set} (MIS) and {\em maximal matching} (MM), which are among the most intensively studied problems in distributed/parallel computing, in \MPC{}.
	
	These problems are known to admit efficient \MPC{} algorithms if the space per machine is near-linear in $n$, the number of vertices in the graph. This space requirement however, as observed in the literature, is often significantly larger than we can afford; especially when the input graph is sparse. In a sharp contrast, in the {\em truly sublinear} regime of $n^{1-\Omega(1)}$ space per machine, all the known algorithms take $\poly\log n$ rounds which is considered inefficient.
	
	Motivated by this shortcoming, we parametrize our algorithms by  the {\em arboricity} $\alpha$ of the input graph, which is a well-received measure of its sparsity. We show that both MIS and MM admit $O(\sqrt{\log \alpha}\cdot\log\log \alpha + \log^2\log n)$ round algorithms using $O(n^\epsilon)$ space per machine for any constant $\epsilon \in (0, 1)$ and using $\widetilde{O}(m)$ total space. Therefore, for the wide range of sparse graphs with small arboricity---such as minor-free graphs, bounded-genus graphs or bounded treewidth graphs---we get an $O(\log^2 \log n)$ round algorithm which exponentially improves prior algorithms.
	
	By known reductions, our results also imply a $(1+\epsilon)$-approximation of maximum cardinality matching, a $(2+\epsilon)$-approximation of maximum weighted matching, and a 2-approximation of minimum vertex cover with essentially the same round complexity and memory requirements.
\end{abstract}

\clearpage

\section{Introduction}
The success of frameworks such as MapReduce \cite{DBLP:conf/osdi/DeanG04, DBLP:journals/cacm/DeanG08}, Hadoop \cite{DBLP:books/daglib/0025439}, or Spark \cite{DBLP:conf/hotcloud/ZahariaCFSS10} has led to a significant interest in better understanding their true computational power. The {\em Massively Parallel Computations} (\MPC{}) model \cite{DBLP:conf/soda/KarloffSV10, DBLP:conf/isaac/GoodrichSZ11, DBLP:conf/stoc/AndoniNOY14, DBLP:journals/jacm/BeameKS17} is arguably the most popular theoretical model that captures the essence of these frameworks while abstracting away their technical details. Compared to traditional parallel or distributed models such as \PRAM{} or \local{}, \MPC{} has advantages such as {\em free} local computation or the possibility of all-to-all communications. In fact, classical parallel algorithms often give rise to \MPC{} algorithms within, asymptotically, the same number of parallel rounds \cite{DBLP:conf/soda/KarloffSV10,DBLP:conf/isaac/GoodrichSZ11}. The main question, however, is whether the advantages of \MPC{} can be leveraged to improve these inherited results. While the answer to this question is clearly positive for a number of problems, it is typically less obvious and highly depends on the problem at hand.

In this paper, we consider two fundamental graph problems of maximal matching and maximal independent set (MIS). While these problems admit trivial sequential greedy algorithms, choosing, in parallel, a subset of vertex-disjoint edges to add to the matching (or a subset of independent vertices to add to the MIS) is non-trivial and requires ``symmetry breaking" between the edges (or vertices) that have similar topologies around them. These problems have been at the heart of parallel/distributed computing from the very early days of the field back in 1980s and have been studied extensively ever since \cite{DBLP:conf/stoc/KarpW84, DBLP:journals/jacm/KarpW85, DBLP:conf/stoc/Luby85, DBLP:journals/ipl/IsraelI86, DBLP:journals/ipl/IsraeliS86, DBLP:journals/jal/AlonBI86, DBLP:conf/focs/Linial87, DBLP:conf/stoc/GoldbergPS87, DBLP:conf/focs/AwerbuchGLP89,DBLP:journals/jacm/BarenboimEPS16, DBLP:conf/soda/Ghaffari16}.

Studying graph problems in the \MPC{} model started with the paper of Karloff et al.~\cite{DBLP:conf/soda/KarloffSV10} who gave $O(1)$ round algorithms for MST and connectivity when the space per machine is $\Theta(n^{1+\delta})$; here $n$ is the number of vertices and $\delta > 0$ is any arbitrary constant. Henceforth, many other graph problems, including maximal matching and MIS, enjoyed $O(1)$ round algorithms in this regime of $\Theta(n^{1+\delta})$ space per machine \cite{DBLP:conf/spaa/LattanziMSV11, DBLP:conf/spaa/KumarMVV13, DBLP:conf/spaa/AhnG15, DBLP:conf/spaa/BehnezhadDETY17, DBLP:conf/spaa/AssadiK17, DBLP:conf/nips/BateniBDHKLM17, harvey2018greedy}. Starting with the breakthrough of \cite{DBLP:conf/stoc/CzumajLMMOS18}, a series of recent papers \cite{DBLP:journals/corr/Coresets, DBLP:journals/corr/ImprovedMPC, DBLP:journals/corr/Konrad}, remarkably, reduced this space requirement to $O(n)$ or even $O(n/\poly\log n)$ while incurring only a slight blow-up of roughly $O(\log \log n)$ on the round complexity. This spectacular progress, however, seems to  inherently depend on the availability of enough space per machine to store nearly all the nodes.

The space requirement of $\widetilde{\Omega}(n)$ is  suitable for {\em dense graphs} where it is mainly the edges of the graph that contribute to its massive size. It, however, defeats the purpose of massive parallelism if the graph is {\em sparse} --- the main focus of this paper --- as one can fit nearly the whole input on one machine!\footnote{We note that these challenges faced for sparse graphs are in a way reminiscent of that of  other big-data settings such as {\em streaming} or {\em sublinear} algorithms.}  This is unfortunate, since many real-world large-scale graphs, such as social networks, tend to be sparse \cite{DBLP:conf/stoc/CzumajLMMOS18}. The most interesting set of parameters for sparse graphs is the {\em truly sublinear} regime of  $O(n^{\epsilon})$ space per machine where $\epsilon < 1$ is a constant.

Adapting known \MPC{} algorithms for maximal matching \cite{DBLP:conf/spaa/LattanziMSV11} or MIS \cite{DBLP:journals/corr/ImprovedMPC, DBLP:journals/corr/Konrad, harvey2018greedy} to the truly sublinear regime offers no benefit. In fact the round complexity of all these algorithms blows, at least, up to $\Omega(\log n)$ --- a bound that  can also be achieved by simulating three decades old algorithms of \cite{DBLP:conf/stoc/Luby85, DBLP:journals/jal/AlonBI86, DBLP:journals/ipl/IsraelI86}. However, going back to the main motivation for considering the truly sublinear regime, which was the case of sparse graphs, it is natural to ask:

\begin{quotebox}
	Can we take advantage of the sparsity of the input graph to improve inherited \PRAM{}/\local{} algorithms in the truly sublinear regime of the \MPC{} model?
\end{quotebox}

\vspace{-0.6cm}
\paragraph{Parametrizing by arboricity.} To address the question above, we initiate the study of truly sublinear \MPC{} algorithms that are parametrized by {\em arboricity} of the input graph. The arboricity of a graph is the minimum number of forests into which its edges can be partitioned. Equivalently, Nash-Williams~\cite{nash1964decomposition} showed that it can be defined as the density of the densest subgraph.\footnote{More precisely, arboricity can be defined as $\max_{S \subseteq V, |S|\geq 2} \lceil |E(S)|/(|S|-1) \rceil$ where $E(S)$ denotes the set of edges between the vertices in $S$.} Arboricity is a well-received measure of sparsity that does not impose strict structural constraints such as planarity, bounds on maximum degree, or the like \cite{DBLP:journals/dc/BarenboimE10, DBLP:conf/focs/BarenboimEPS12, DBLP:conf/soda/EsfandiariHLMO15, DBLP:conf/esa/CormodeJMM17, DBLP:journals/corr/GrigorescuMZ16}. Indeed most families of sparse graphs, including graphs that exclude a fixed minor (such as planar graphs), graphs of bounded genus, bounded degree, bounded treewidth, or pathwidth have all constant arboricity. Furthermore, graphs with constant arboricity may also have a genus of up to $O(n)$ or have $K_{\sqrt{n}}$ as a minor. We note that none of our algorithms assume arboricity is bounded by a constant.

We show that both maximal matching and MIS can be solved in $O(\sqrt{\log \alpha}\log \log \alpha + \log^2\log n)$ rounds where $\alpha$ denotes arboricity. Remarkably, our algorithms do not require to be given the arboricity of the graph. Since arboricity may be up to $n$, this bound still requires $O(\log n)$ rounds in the general case. However, for graphs with a moderately smaller arboricity, it implies an exponential improvement over the round complexity of inherited algorithms. For instance, Barenboim et al.~\cite[Theorem 7.7]{DBLP:conf/focs/BarenboimEPS12} show,  by adapting the celebrated lower bounds of Kuhn et al.~\cite{DBLP:journals/jacm/KuhnMW16}, that even to compute a maximal matching of trees --- which by definition, have arboricity only 1 --- any \local{} algorithm provably requires $\Omega(\sqrt{\log n})$ rounds.

\smparagraph{Comparision with graph connectivity.} It is worth noting that some known hard inputs for graph connectivity in the truly sublinear regime for which no $o(\log n)$ round algorithm is known (and is, in fact, conjectured to not exist \cite{DBLP:conf/spaa/RoughgardenVW16, DBLP:conf/stoc/AndoniNOY14}) have $O(1)$ arboricity. One example is the so called {\em one-cycle vs two-cycle} problem where we are promised that the input is composed of either two cycles or one cycle. Here, the arboricity of the input graph is only 2. It is therefore perhaps surprising that the seemingly harder problems of maximal matching and MIS can be solved exponentially faster for such graphs. In comparison, when the space per machine is $\widetilde{\Theta}(n)$, graph connectivity can be solved in $O(1)$ rounds \cite{DBLP:conf/soda/Jurdzinski018, DBLP:journals/corr/abs-1802-10297, DBLP:journals/corr/abs-1805-02974} but the fastest algorithms known for MIS and (approximate) matching take $O(\log \log n)$ rounds \cite{DBLP:conf/stoc/CzumajLMMOS18, DBLP:journals/corr/Coresets, DBLP:journals/corr/ImprovedMPC, DBLP:journals/corr/Konrad}.

\subsection{The \MPC{} Model}
We consider the most restrictive variant of the {\em Massively Parallel Computations} (\MPC{}) model which was initially introduced by \cite{DBLP:conf/soda/KarloffSV10} and further refined by \cite{DBLP:conf/isaac/GoodrichSZ11, DBLP:conf/pods/BeameKS13, DBLP:journals/jacm/BeameKS17, DBLP:conf/stoc/AndoniNOY14}. An input of size $N$ is initially distributed among $M$ machines, each having a local space of size $S$. Computation proceeds in synchronous rounds: Within each round, each machine performs a local computation on its data and at the end communicates with other machines. The only restriction on the communications is that the total size of the messages sent or received by each machine should not exceed its memory. We desire algorithms with a substantially sublinear space of $S = N^{1-\Omega(1)}$ per machine and ideally only enough total space to store the input, i.e., $S \cdot M = \widetilde{O}(N)$. Moreover, we are interested in algorithms that can be adjusted to use a local space of size $S = O(N^{\epsilon})$ for any constant $\epsilon \in (0, 1)$.

For graph problems, the input graph $G=(V, E)$ with $n$ vertices and $m$ edges is given as follows: The edges, which are pairs of their endpoints' IDs, are initially distributed (adversarially) among the machines; thus, the input size is $O(m)$. Moreover, the space per machine is assumed to be $O(n^{\epsilon})$ for any desirably small constant $\epsilon \in (0, 1)$.

\subsection{Further Related work}
\noindent\textbf{MIS on trees.} The most relevant to our work, is the paper of Brandt, Fischer and Uitto~\cite{mistreesmpc} in which they design an $O(\log^3\log n)$ round algorithm to find MIS of trees in the same \MPC{} setting. Their algorithm is based on a clever {\em subsampling} idea. They use structural properties of trees to show that if we sample the edges uniformly at random with an appropriate probability $p$, then the tree is decomposed into small subtrees of diameter at most $O(\log_{1/p} n)$ and size at most $O(n^\epsilon)$. They then gather each subtree into a machine in $O(\log \log n)$ rounds and find an MIS on it, which they show reduces the maximum degree of the main graph by a factor of $\Delta^{\Omega(1)}$. This means that only $O(\poly \log \log n)$ iterations of this procedure is sufficient to make max degree polylogarithmic where known \local{} algorithms can be employed to solve the problem in $O(\log \log n)$ rounds. 

We were able to generalize the algorithm of Brandt et al. to solve maximal matching on trees in  $O(\poly \log \log n)$ rounds as well. However, subsampling does not preserve the above-mentioned characteristics beyond trees, even when the arboricity is 2. For example, the argument that shows subsampling reduces diameter to $O(\log_{1/p} n)$ is based on the fact that there are only $O(n^2)$ paths in trees and each path of length at least $3 \log_{1/p} n$ is completely sampled with probability $p^{3 \log_{1/p} n} = 1/n^3$ (thus we can use union bound). However, even on a grid, which has arboricity 2, we may have exponentially many paths and the argument above breaks down. In fact, one can construct a delicate input with arboricity $O(1)$ that has $n^{\Omega(1)}$ vertices of degree at least $n^{\Omega(1)}$ where subsampling leads to connected components that do not fit the memory of a single machine.

\smparagraph{\PRAM{}/\local{} algorithms.} As mentioned above, traditional parallel algorithms that are not extremely resource heavy can be seamlessly simulated within asymptotically the same number of rounds in \MPC{} \cite{DBLP:conf/soda/KarloffSV10,DBLP:conf/isaac/GoodrichSZ11}. Here we briefly overview known results in these settings to show that our algorithms indeed use the ``full power" of \MPC{} to improve them. On the \PRAM{} model, algorithms of Luby~\cite{DBLP:conf/stoc/Luby85} and Israeli and Itai~\cite{DBLP:journals/ipl/IsraelI86} can be used to solve MIS and maximal matching in $O(\log n)$ rounds. This is however much larger than our running time of $O(\log \alpha + \log^2 \log n)$ if $\alpha \leq n^{o(1)}$.    On the \local{} model, for graphs of arboricity $\alpha$, algorithms of \cite{DBLP:conf/soda/Ghaffari16, DBLP:journals/jacm/BarenboimEPS16} respectively solve MIS and maximal matching in $O(\log \alpha + \sqrt{\log n})$ rounds. As mentioned before, these bounds are tight at least for maximal matching due to the $\Omega(\sqrt{\log n})$ lower bounds of \cite{DBLP:journals/jacm/KuhnMW16, DBLP:conf/focs/BarenboimEPS12} on unrooted trees which have arboricity 1. Our algorithms improve these bounds significantly if $\alpha \leq o(2^{\sqrt{\log n}})$ and in fact exponentially if $\alpha \leq \poly \log n$. We note that there are also a handful of faster \PRAM{}/\local{} algorithms for special cases. For instance, if the input graph is a rooted tree, its MIS can be solved in $O(\log^* n)$ rounds of \PRAM{}/\local{} \cite{DBLP:conf/stoc/ColeV86}. We refer to \cite{DBLP:journals/jacm/BarenboimEPS16, DBLP:series/synthesis/2013Barenboim} for a more thorough overview of known results in these settings.

\smparagraph{Other big-data settings.} The intricacy designing graph algorithms in the truly sublinear regime with $n^{1-\Omega(1)}$ local space, where $n$ denotes the number of vertices, also extends to other ``big-data" models such as the {\em streaming} setting. There has been a long line of research in estimating the size of maximum matching, particularly in graphs of bounded arboricity using sublinear in $n$ space in the streaming setting (see e.g., \cite{DBLP:conf/soda/EsfandiariHLMO15, DBLP:conf/soda/ChitnisCEHMMV16, DBLP:conf/approx/McGregorV16, DBLP:conf/esa/CormodeJMM17, DBLP:conf/soda/AssadiKL17, DBLP:conf/soda/AssadiKL17, DBLP:conf/soda/0001V18} and the references therein).

\smparagraph{Concurrent work.} In an independent and concurrent work, Brandt, Fischer, and Uitto \cite{concurrent} also consider maximal matching and MIS on low arboricity graphs in the truly sublinear regime of \MPC{}. The round complexity and memory requirements of both works are essentially the same. The main technical ingredient of both results is an $O(\log^2 \log n)$ round algorithm that reduces maximum degree to $\poly(\alpha, \log n)$ implying  $O( T\big(\poly(\alpha, \log n)\big) + \log^2\log  n )$ round algorithms for MIS or maximal matching where $T(d)$ denotes the number of rounds required to solve these problems on a graph with maximum degree $d$.\footnote{We note that the round complexity of our algorithm (as well as that of \cite{concurrent}) can also be expressed as a function of $\Delta$ to be $O(T(\poly (\alpha, \log n)) + \log \log \Delta \cdot \log \log n)$. For clarity purposes, our main results are only expressed as functions of $\alpha$ and $n$.}  It is shown in both papers that $T(d) \leq O(\log d)$ by simulating algorithms of \cite{DBLP:conf/soda/Ghaffari16, DBLP:conf/focs/BarenboimEPS12}, meaning that the round complexity is $O(\log \alpha + \log^2\log n)$. Another concurrent work by Ghaffari and Uitto \cite{sqrtlogd} quadratically improves the bound on $T(d)$. Using this as a black-box, the round complexity of our algorithms as well as those of Brandt et al. can be improved to $O(\sqrt{\log \alpha} \cdot \log\log \alpha + \log^2\log n)$.

\section{Technical Overview}

Our main result is what follows. In the sequel we overview the main intuitions in achieving it.

\newcommand{\mainthm}[0]{
	For any given graph $G$ with $n$ vertices, $m$ edges, and arboricity $\alpha$, and for any desirably small constant $\epsilon \in  (0, 1)$, there exists an algorithm that with high probability\footnote{As standard, {\em with high probability} indicates probability at least $1-n^{-c}$ for any desirably large constant $c$.} computes a maximal independent set (or maximal matching) of $G$ in $O(\sqrt{\log \alpha}\cdot \log\log \alpha + \log^2 \log n)$ rounds of \MPC{} using $\O{n^{\epsilon}}$ space per machine and $\widetilde{O}(m)$ total memory.
}

\begin{graytbox}
\begin{theorem}\label{thm:main}
	\mainthm
\end{theorem}
\end{graytbox}

\begin{remark}
The algorithm for Theorem~\ref{thm:main} does not require to be given $\alpha$. Furthermore, for all graphs of arboricity up to $\poly \log n$, even if $\epsilon = 1/\poly\log \log n$ (i.e., the space per machine is mildly sub-polynomial) the algorithm takes only $O(\poly \log \log n)$ rounds with high probability.	
\end{remark}

\begin{remark}\label{rem:corollaries}
	Using known reductions, by employing our maximal matching algorithm, a $(1+\epsilon)$-approximation for {\em maximum matching} \cite{DBLP:journals/corr/Coresets, DBLP:conf/approx/McGregor05}, a $(2+\epsilon)$-approximation for {\em maximum weighted matching} \cite{DBLP:journals/siamcomp/LotkerPR09}, and a $2$-approximate {\em vertex cover} can be obtained in asymptotically the same number of rounds of \MPC{} with the same memory requirements given that $\epsilon$ is any arbitrarily small constant.
\end{remark}

As observed by \cite{mistreesmpc}, the insufficiency of space to store all the vertices in one machine imposes challenges similar to those faced by algorithms in the \local{} \cite{PelegBook} model: There is one processor on each of the nodes of the input graph and two processors can communicate in each round if and only if there is an edge between their corresponding vertices. The fact that the vertices, in the truly sublinear regime of \MPC{}, have to make decisions (such as joining the MIS) based solely on a {\em small} neighborhood that they observe around them, makes the algorithmic challenges of the two models similar. We need to keep in mind, however, that the constraints that impose such locality in the two models are fundamentally different. Roughly, in \local{}, the diameter of the subgraph that each vertex can observe is small but in \MPC{}, it is the size of this subgraph that is restricted to be sublinear. Neither of the two subsumes the other, a graph with small diameter may have a large size and a sublinear size graph may have a large diameter.

However, a key difference between the two models that makes us hope for faster \MPC{} algorithms is the possibility of all-to-all communications. To illustrate this over a simple example, consider a directed path $(v_1, v_2, \ldots, v_d)$. It is not hard to see that in the \local{} model, it takes at least $d - 1$ rounds for $v_1$ to send one bit of message to $v_d$. However, thanks to all-to-all communications, it can be done in only $O(\log d)$ rounds of \MPC{} using the well-known {\em pointer jumping} technique: Initially, for any $i < d$, set $p(v_i) := v_{i+1}$  and in each round update it to be $p(v_i) \gets p(p(v_i))$. In only $O(\log d)$ rounds $p(v_1)$ will point to $v_d$. This is possible since vertex $v_i$ can directly communicate with vertex $u = p(v_i)$ and ask for the value of $p(u)$. Achieving such exponential improvements, however, is typically much more intricate for other problems due to the space restrictions of \MPC{}. For readers familiar with the {\em congested clique} \cite{PelegBook} model, we note that the availability of all-to-all communications there also allows for such improvements. However, congested clique is much stronger than \MPC{} with sublinear in $n$ space. In fact, congested clique is almost equivalent to the variant of \MPC{} with $\Theta(n)$ space per machine \cite{DBLP:journals/corr/abs-1802-10297}.

To further demonstrate the relevance of the above {\em exponential growth} idea to our problems, we recall a beautiful (and  well-known) property of \local{} algorithms. In any $r$-round \local{} algorithm, the final state of each node/edge is merely a function of its $r$-hop (i.e., the nodes/edges that are at distance at most $r$). This has been extensively  used in the literature to prove lower bounds, but has also given rise to a few algorithmic ideas (see e.g., \cite{DBLP:journals/tcs/ParnasR07, DBLP:conf/soda/AlonRVX12} and the follow-up work or \cite{DBLP:conf/podc/Ghaffari17}). Combined with the $O(\log n)$ round \local{} algorithm of Luby~\cite{DBLP:conf/stoc/Luby85} for MIS, or that of Israeli and Itai~\cite{DBLP:journals/ipl/IsraelI86} for maximal matching, this property implies that if in \MPC{}, we manage to collect the $O(\log n)$-hop of each vertex in a machine responsible for it, we can locally simulate these algorithms without any further communications.\footnote{We note that since these algorithms are randomized, one also needs to collect the tape of random bits of each vertex as well so that the results computed on different machines are compatible.} Using the exponential growth idea, we hope to be able to do this in much faster than $O(\log n)$ rounds. There are however two fundamental barriers for this:

\begin{description}
	\item[Local memory barrier.] The $O(\log n)$-hop of a vertex may be as large as $\Omega(m)$, exceeding the local space of a machine. Even the 1-hop of a vertex with degree higher than $\omega(n^{\epsilon})$ cannot be stored in one machine.
	\item[Global memory barrier.] Storing the neighborhood of each vertex on its corresponding machine leads to multiple copies of each vertex and thus a total aggregated memory of significantly larger than the input size, $m$. 
\end{description}

Let us forget the global memory barrier for now (which actually turns out to be an important restriction) and focus on handling the local memory problem. Denote the maximum degree of the graph by $\Delta$ and set $\delta = \epsilon/3$. We can safely assume for $t=\delta \log_{\Delta} n$, that the $t$-hop of every vertex fits the memory of one machine since $\Delta^{\delta \log_\Delta n} = n^\delta < O(n^{\epsilon})$. This implies that we can indeed simulate $t$ rounds of a \local{} algorithm in one round if we first collect the $t$-hops (which we show takes only $O(\log t)$ rounds). However, $t$ is usually smaller than the actual number of rounds that the algorithm takes. A way to overcome this is to share the {\em states}. That is, having the state of each vertex by the end of round $t$, we can share these states with all other machines in one round of communication and simulate the next $t$ rounds of the algorithm to obtain the states by round $2t$. We can repeat this to simulate $r$ rounds of our \local{} algorithm in $O(r/t + \log t)$ rounds. Henceforth, we call this technique {\em blind coordination}. Note that for this idea to work, the states of the \local{} algorithm have to be crucially small so that they can be shared and stored on the machines. However, even incorporating blind coordination does not help when the maximum degree is large. For instance, when $\Delta = \Omega(n^\epsilon)$, even the 1-hop of a vertex may not fit the memory of a single machine, meaning that blind coordination does not lead to any improvements. This implies that the main challenge, similar to many other known \MPC{} algorithms, is to reduce the the maximum degree of the graph.

For ease of exposition and to convey the intuitions, we assume in this section that the arboricity of the input graph is $O(1)$. We borrow a subroutine first introduced by Barenboim et al.~\cite[Theorem 7.2]{DBLP:conf/focs/BarenboimEPS12} for the \local{} model and use it in a novel way to reduce the degree to our desired bound. This algorithm, with slight modifications, guarantees that for any $\tau \geq \log^{O(1)} n$ (that is also sufficiently larger than arboricity,) one can reduce the maximum degree to $\tau$ in $O(\log_{\tau} n)$ rounds by committing a subset of edges (or vertices) to the maximal matching (or MIS).\footnote{Since we assume that arboricity is constant in this section, we have hidden the actual dependence of the running time on the arboricity.} Call a vertex $v$ {\em high-degree} if $\deg(v) > \tau$ and {\em low-degree} otherwise. This round complexity is achieved since the algorithm removes $\tau^{\Omega(1)}$ fraction of high-degree vertices in each round by matching them to their low-degree neighbors (or by adding their low-degree neighbors to MIS). The algorithm turns out to be very simple to implement and intuitive. For instance for maximal matching, in each round, after discarding a subset of edges, each low-degree vertex proposes to one of its high-degree neighbors uniformly at random and then each high-degree vertex gets matched to one of its proposing neighbors (if any) arbitrarily.\footnote{See Algorithm~\ref{alg:degreereduction} for the formal statement.} 

The intuition behind the analysis of this subroutine is roughly as follows: Fix a high-degree vertex $v$ and suppose it is likely to survive $\ell$ rounds and remain high-degree. For this to happen, not only almost all neighbors of $v$ have to be high-degree, but the neighbors of its neighbors should also be high-degree and this should continue for roughly $\ell$ levels. Due to the small arboricity of the graph, these high-degree vertices cannot be highly inter-connected (otherwise we have a dense subgraph) and thus each level requires $\tau^{\Omega(1)}$ additional nodes. Therefore, $\ell$ cannot exceed $O(\log_\tau n)$.

This subroutine helps in finding maximal matching or MIS in $O(\sqrt{\log n})$ rounds of \local{} when the graph has a small arboricity. Without delving into details, this is achieved by setting $\tau = O_\alpha(2^{\sqrt{\log n}})$ and then using another algorithm on the remaining lower degree graph.

To use the advantages of \MPC{} to improve over this bound exponentially, instead of assigning a fixed value to $\tau$ and using the subroutine in one shot, we iteratively assign different values to $\tau$ and combine it with the blind coordination lemma described above. More precisely, we divide the algorithm into $O(\log \log n)$ phases (not rounds) that in turn reduce the maximum degree by a polynomial factor from $\Delta$ to $\sqrt{\Delta}$ until it eventually becomes desirably small. This iterative process, intuitively, helps in the following way: if we are in a phase where the maximum degree of the graph is large, say $\Omega(n)$, reducing it to $\sqrt{n}$ takes only $O(\log_{\sqrt{n}} n) = O(1)$ rounds. Therefore, we can afford to directly simulate the algorithm in \MPC{} without any round compressions. Moreover, when the maximum degree gets smaller into a point where $O(\log_{\Delta} n)$ becomes the bottleneck, we can use the blind coordination procedure which precisely works well when the maximum degree is small. There are $O(\log \log n)$ phases, each takes at most $O(\log \log n)$ rounds to simulate (due to blind coordination); thus the algorithm overall takes only $O(\log^2 \log n)$ rounds.

Another nice property of iteratively changing the thresholds that we set for $\tau$ is that we do not require to know the arboricity as opposed to the above-mentioned \local{} algorithms. The reason is that, once we reach an unsuccessful phase, i.e., a phase where the maximum degree is not reduced to the desired bound (which is easy to check in \MPC{}), it is w.h.p. guaranteed to be of size $\poly \alpha \cdot \poly \log n$, thus we can terminate the future phases and switch to the finish-up phase.

\smparagraph{Optimizing the global memory.} The challenge in optimizing the global memory is mainly centered around the blind coordination procedure which we used extensively in the above algorithm. This actually turns out to be a rather serious problem and we are not aware of any way to generally apply blind coordination without using $n^{1+\Omega(\epsilon)}$ total space which may be much larger than $m$, the input size. To illustrate this, we first show why a natural idea does not work and then proceed to show how we optimize total memory using specific properties of our algorithms.

At the first glance, it seems extremely wasteful to store the $t$-hop (recall that $t=\delta \log_{\Delta} n$ and $\delta = \epsilon/3$) of  {\em every} vertex to simulate $t$ rounds of a \local{} algorithm. Indeed having the $t$-hop of a vertex $v$, implies that not only we can compute the state of $v$ after $t$ rounds, but also implies that we can compute that of its direct neighbors after $t-1$ rounds, since their $(t-1)$-hop is also included in this subgraph, and so on. One may wonder whether it is possible to compute all the states after $\Omega(t)$ rounds by collecting only the $t$-hops of only a subset of the vertices. Unfortunately, such ideas do not generally work and to compute the state of every vertex after $d \cdot t$ rounds, one can construct a graph on which we inevitably need $n^{1+\Omega(d \cdot t)}$ total space.

Here we only highlight the intuitions that lead to bypassing the barrier mentioned above. Suppose that our goal is to apply blind-coordination to simulate $t$ rounds of a \local{} algorithm. Our main intuition is that if we can manage to show structurally that the state of a vertex $v$ is finalized by some round $i \ll t$ and does not change afterward, then having the $i$-hop of $v$ suffices to simulate the algorithm by round $t$. To show a simple concrete example, recall that in each phase of our algorithm we reduce the maximum degree from $\Delta$ to $\sqrt{\Delta}$. Within each phase, a low-degree vertex $v$ (i.e., $\deg(v) < \sqrt{\Delta}$) whose all neighbors are also low-degree, is completely ignored by the algorithm. Therefore, we do not need to collect a large neighborhood around this vertex to simulate the algorithm by the end of the current phase. Complications arise since ignoring these vertices may not release enough space to collect the $t$-hop of other vertices. For instance, it could be the case that a large fraction of the vertices are indeed high-degree. Total memory management in such scenarios turns out to be much more challenging. We have to adaptively detect high-degree vertices whose states are finalized by simulating a few rounds and then stop growing the regions around them. We show that with careful analysis and adjustments to the algorithm, total memory can be reduced from $\widetilde{O}(m+n^{1+\Omega(\epsilon)})$ to $\widetilde{O}(m)$ while keeping the round complexity asymptotically the same.

\section{Basic Algorithmic Tools for \MPC{}}
In this section we describe a set of basic algorithmic primitives for graph problems in the \MPC{} model.

\subsection{Load Balancing}

Throughout the paper, for different applications, we encounter the following problem: A number is written on each of the vertices of the graph, and for every vertex, we need to compute a function of the numbers written on its neighbors. The simplest case is finding the degree of each vertex where the numbers are simply one and the function is sum. Another example is finding the minimum label written on the neighbors of each vertex to break symmetry. The problem is that if a vertex has degree higher than the space per machine, we are not able to store the numbers written on its neighbors in one  machine and the task has to be distributed. We show that simple functions such as max, min, sum, etc., can be computed in $O(1)$ rounds using $\Ot{m}$ total space.

To remain as general as possible, we define {\em separable} functions. All the aforementioned functions such as $\max$, $\min$, sum, etc., are separable.

\begin{definition}\label{def:separable}
Let $f: 2^{\mathbb{R}} \rightarrow \mathbb{R}$ denote a set function. We call $f$ \emph{separable} iff for any set of reals $A$ and for any $B \subseteq A$, we have $f(A) = f\big(f(B), f(A \setminus B)\big)$.
\end{definition} 

The following lemma implies that it is possible to compute the value of a separable function $f$ on each of the vertices in merely $O(1)$ rounds. The proof is a simple application of the well-known {\em balls into bins} problem; thus we defer it to Section~\ref{sec:load-balancing}.

\newcommand{\separablelemma}[0]{
Suppose that on each vertex $v \in V$, we have a number $x_v$ of size $O(\log n)$ bits and let $f$ be a separable function. There exists an algorithm that in $O(1)$ rounds of \MPC{}, for every vertex $v$, computes $f(\{x_u \, | \, u\in N(v) \})$ and with probability at least $1-n^{-c}$ (for any desirably large constant $c$) uses $\O{n^{\epsilon}}$ space per machine and $\Ot{m}$ total space where $\epsilon$ is any desirably small constant in $(0, 1)$.}
\begin{lemma} \label{lem:generalmaxload}
\separablelemma{}
\end{lemma}

We remark that even if $\epsilon$ is sub-constant, Lemma~\ref{lem:generalmaxload} works within $O(1/\epsilon)$ rounds. For ease of exposition, we assume $\epsilon$ is constant throughout the paper unless explicitly stated otherwise. Nonetheless, an extra factor of $1/\epsilon$ appears in the round complexity of our algorithms if $\epsilon$ is sub-constant.

\subsection{Exponential Growth via All-to-All Communication}\label{sec:expgrowth}

The {\em exponential growth} technique allows us to collect the $t$-hop of every vertex in $O(\log t)$ rounds so long as we are guaranteed that the size of each of them is sufficiently small. The idea is to inductively collect the $2^{i}$-hop of every vertex by round $i$. We note that similar techniques have been used in the literature under different names such as broadcasting, adding 2-hops, etc. \cite{mistreesmpc, DBLP:journals/corr/abs-1805-03055, DBLP:conf/podc/Ghaffari17}

\begin{lemma}\label{lem:collecting}
Given that for any vertex $v$, size of its $t$-hop is bounded by $n^{\beta}$ for any $\beta \leq \epsilon /2$, there exists an algorithm that gathers the $t$-hop of every vertex in at least one machine within at most $O(\log t)$ rounds of \MPC{} using $O(n^{\epsilon})$ space per machine and $O(n^{1+2\beta})$ total space.
\end{lemma}

\begin{proof}
We first assign vertices to machines such that any of the $\Theta(n^{1+2\beta-\epsilon})$ machines is responsible for at most $k = O(n^{\epsilon -2\beta})$ vertices. This can easily be done, e.g., by making machine number $i$ responsible for vertices with ID in $\{(i-1)k+1, \ldots, ik\}$. Note that each machine has enough space to store data of size $\O{n^{2\beta}}$ for any vertex that it is responsible for since $k\cdot n^{2\beta} = n^{\epsilon}$.
	Therefore, it only remains to collect the $t$-hop into each machine in $O(\log t)$ rounds. The algorithm is what follows: In round 1, for any vertex $v$, each edge incident to $v$ is sent to the machine responsible for $v$. We call the set of all these edges $\mathcal{N}(v)$. Then iteratively for $O(\log t)$ rounds, each machine, for any vertex $v$ that it is responsible for, and for any vertex $u \in \mathcal{N}(v)$, requests $\mathcal{N}(u)$ from the machine responsible for $u$ and updates $\mathcal{N}(v)$ to be $\mathcal{N}(v) \gets \cup_{u \in \mathcal{N}(v)} \mathcal{N}(u)$. Throughout the algorithm, we further ensure that for each vertex $v$, $\mathcal{N}(v)$ only contains the edges in its $t$-hop. This can be simply checked within each machine.
	
	One can easily confirm that by the end of iteration $i+1$, $\mathcal{N}(v)$ contains the vertices in the $2^{i}$-hop of vertex $v$, therefore after $O(\log t)$ iterations of the algorithm the vertices in the $t$-hop of every vertex is stored in the machine responsible for it.
	%Also, the edges in the $t$-hop of $v$ can be simply gathered by sending a requests the machines responsible for its vertices.
	It remains to show that the algorithm does not violate the messages limits and the space restrictions. Each machine, as argued above, is responsible for only $O(n^{\epsilon - 2\beta})$ vertices. Each of these vertices will have a $t$-hop of size at most $n^{\beta}$. Therefore, the collection of the $t$-hops of all these vertices has size at most $O(n^{\epsilon - \beta})$. It might happen that in the final round, we request the $t$-hop of each vertex collected in a machine, but since each of them sends a subgraph of size $O(n^\beta)$, the total size of messages received by each machine is at most $O(n^{\epsilon - \beta} \cdot n^\beta) = O(n^{\epsilon})$. A similar argument shows that no machine sends more than $O(n^\epsilon)$ messages. Since we have $O(n^{1+2\beta - \epsilon})$ machines, each with a local space of size $O(n^{\epsilon})$, the total space is $O(n^{1+2\beta})$.	
%	\begin{tboxalg}{Collecting the $t$-hop neighborhood of each vertex in its responsible machine.}\label{alg:collecting}
%		\begin{enumerate}[label={(\arabic*)}, topsep=0pt,itemsep=0ex,partopsep=-1ex,parsep=1ex,leftmargin=*]
%			\item In round 1, for each vertex $v$, all of its incident edges are sent to the machines that is responsible for it.
%			\item For any vertex $v$ we initially set $\mathcal{F}_1(v) \gets N(v)$, $\mathcal{N}_1(v) \gets N(v)$.
%			\item For $i \in 1\ldots \lceil \log t \rceil$ rounds, we do the following:
%			 	\begin{enumerate}
%			 		\item Each machine, for each vertex $v$ that it is responsible for, and for any vertex $u \in \mathcal{F}_{i}(v)$, sends a message of ``need neighborhood of $u$" to the machine responsible for $u$.
%			 		\item Each machine, for any received message of ``need neighborhood of $u$", sends the set $\mathcal{N}_{i}(u)$ to the machine that sent this request.
%			 		\item Each machine for each vertex $v$ that it is responsible for, sets $$\mathcal{N}_{i+1}(v) \gets \mathcal{N}_{i}(v) \cup \Big(\bigcup_{u\in \mathcal{F}_{i}(v)} \mathcal{N}_{i}(u)\Big), \qquad \mathcal{F}_{i+1}(v) \gets \mathcal{N}_{i+1}(v) \setminus \mathcal{N}_{i}(v).$$
%			 	\end{enumerate}
%		\end{enumerate}
%	\end{tboxalg}
\end{proof}

%\begin{lemma}
%	Let $\epsilon'$ be $\min\{\epsilon, \delta\}/3$ and define $t$ to be $\lfloor \epsilon' \log_\Delta n \rfloor$. There exists a low-memory MPC algorithm that in $O(\log \log_{\Delta} n)$ rounds, collects the $t$-hop of every vertex in the machine responsible for it.
%\end{lemma}

\section{Blind Coordination}
The goal of this section is to highlight the simple but powerful concept of {\em blind coordination} that we use extensively in the forthcoming sections. We apply this technique to compress multiple rounds of a large class of \local{} algorithms, which we call {\em state-congested} local algorithms in much fewer number of rounds of \MPC{}. Roughly speaking, in a state-congested local algorithm, we can maintain {\em states} on the vertices/edges over the rounds of the algorithm, but we restrict these states to be of size $O(\log n)$ bits and be dependent (loosely speaking) only on the states of the 1-hop of every vertex/edge at the previous round.

\begin{definition}\label{def:reclocal}
	A distributed \local{} algorithm is {\em state-congested} if:
	\begin{enumerate}
		\item By the end of each round $r$, on any node $v$ (and respectively on any edge $e$), the algorithm stores a state $s_r(v)$ (resp. $s_r(e)$) of size $O(\log n)$ bits. The initial state $s_0(v)$ of each vertex $v$ is its ID and the initial state $s_0(e)$ of each edge $e$ is the IDs of its two endpoints.
		\item The state $s_r(v)$ of each node $v$ by the end of any round $r$, depends only on its state $s_{r-1}(v)$ in the previous round, the states of its incident edges $\{s_{r-1}(e)\, | \, e \ni v \}$ in the previous round, and its tape $\rho(v)$ of $\poly\log n$ random bits. Furthermore, the state $s_r(e)$ of each edge $e=(u, v)$ by the end of round $r$ is only a function of $s_{r-1}(u)$,  $s_{r-1}(v)$ and $s_{r-1}(e)$.

		\item The states of the vertices/edges at the last round of the algorithm are sufficient in determining, collectively, the output of the algorithm.
	\end{enumerate}
\end{definition}

The key property of state-congested local algorithms is that the intermediate states of the algorithm are also small. This is in contrast, for example, with algorithms in which each vertex $v$ first collects its, say, $O(\log n)$-hop and then makes its final decision in one shot. We note that state-congested local algorithms are similar to, but more restrictive, than a variant of \local{} algorithms called \congest{} (see \cite{PelegBook}) where the messages over the links are restricted to have $O(\log n)$ bits. Before describing the main result of this section, we need another definition.

\begin{definition}
	We call a state-congested local algorithm {\em low-memory}, if updating the state each node $v$ can be done in a space of size $O(\deg(v) \cdot \poly\log n)$ bits and updating the state of each edge $e$ requires a space of size $O(\poly \log n)$ bits.
\end{definition}

The definition above is required to ensure, e.g., that once we have the states and random tapes of all neighbors of a node $v$, we can update the state of $v$ without using any extra space. This is almost always satisfied. 

We are now ready to formalize the main lemma of this section which results in compressing a state-congested local algorithm in much fewer number of rounds of a low-memory \MPC{} if the maximum degree $\Delta$ of the graph is small. The main theorem is as follows:

\begin{lemma}\label{lem:blindcoordination}
	For any graph with $n$ vertices, $m$ edges, and maximum degree $\Delta \leq n^{\epsilon}$, where $\epsilon$ is a desirably small constant number in $(0, 1)$, one can compress $r$ rounds of any low-memory state-congested local algorithm in $O\big(\frac{r}{\log_\Delta n} + \log \log_\Delta n\big)$ rounds of \MPC{} using $O(n^{\epsilon})$ space per machine and $O(m+n^{1+2\epsilon/3})$ total space.
\end{lemma}

\begin{proof}
Suppose that our goal is to compress $r$ rounds of a low-memory state-congested local algorithm $\mathcal{A}$. Initially, each vertex will be assigned to a machine that will be responsible for keeping track of its state. Note that since the total space is at least $n^{1+2\epsilon/3}$, and memory per machine is $O(n^{\epsilon})$, we have at least $\Omega(n^{1-\epsilon/3})$ machines. It suffices to make each machine responsible for $n^{2\epsilon/3}$ vertices. This assignment can easily be done based on, say, the vertices' IDs. Let $t := \lfloor\frac{\epsilon}{3} \log_\Delta n \rfloor$. We first  collect the $t$-hop of each vertex in the machine that is responsible for it. Note that $t$ is chosen to be small enough that the $t$-hop of every vertex has at most $\Delta^{t} \leq n^{\epsilon/3}$ edges which is substantially smaller than the memory per machine. Therefore, we can use the exponential growth algorithm of Lemma~\ref{lem:collecting} to collect the $t$-hop of every vertex in the machine responsible for it in only $O(\log t) = O(\log \log_\Delta n)$ rounds using a total memory of size $\O{m+n^{1+2\epsilon/3}}$.

After collecting the neighborhoods, on each machine we run $t$ rounds of $\mathcal{A}$ on the subgraph that is stored in it. This can be done in only one round of \MPC{} since no communication between the machines is required. The main intuition behind the {\em blind-coordination} idea is that the final state of a vertex $v$ in the machine that is responsible for it is exactly the same as that of $v$ after $t$ rounds of the original algorithm $\mathcal{A}$. We emphasize that a vertex $v$ may also be stored in machines not responsible for $v$, and in fact, the states computed for $v$ in those machines might be completely different from its {\em correct} state by the end of round $t$ of $\mathcal{A}$. However, crucially, the state of each vertex matches its correct state in the machine responsible for it. Formally, let us denote by $s_i(.)$ the state of a vertex or an edge by the end of round $i$ of algorithm $\mathcal{A}$ and denote by $\hat{s}_{\mu, i}(.)$, for any $i \leq t$, the state of a vertex or an edge after simulating $i$ rounds of $\mathcal{A}$ on the subgraph stored in machine $\mu$; we have:

%\begin{corollary}\label{cor:1}
%		Let $\mu$ be the machine responsible for a vertex $v$. We have $\hat{s}_{\mu, 2t}(v) = s_{2t}(v)$. Moreover, for any edge $e \ni v$, we have $\hat{s}_{\mu, 2t-1}(e) = s_{2t-1}(e)$.
%\end{corollary}
%
%As a corollary of the following lemma is Corollary~\ref{cor:1}.

\begin{observation}\label{obs:linial}
Let $v$ denote an arbitrary vertex whose $i$-hop is stored in machine $\mu$ and let $\beta$ be an arbitrary non-negative integer. If for any vertex or edge $x$ in the $i$-hop of $v$ we have its correct state by round $\beta$, i.e., $\hat{s}_{\mu, \beta}(x) = s_{\beta}(x)$, then we compute the correct state of $v$ after $\beta +i$ rounds in machine $\mu$ , i.e., $\hat{s}_{\mu, \beta+i}(v) = s_{\beta+i}(v)$ without any round of communication. Similarly, for any edge $e$ incident to $v$, we have $\hat{s}_{\mu, \beta+i}(e) = s_{\beta+i}(e)$. 
\end{observation}
\begin{proof}
	We simply prove this by induction on $i$. For $i=1$, since we have the 1-hop of $v$ and the state of edges incident to $v$ are correct by the end of round $\beta$, the machine computes the correct state $\hat{s}_{\mu, \beta + 1}$ for $v$ by definition of state-congested local algorithms. For larger values of $i$, having the $i$-hop of $v$ in machine $\mu$ implies that we also have the $(i-1)$-hop of its neighbors and the state computed for them after $\beta+i-1$ rounds matches $s_{\beta + i -1}$. Thus, in the next step, we correctly compute the state of $v$ after $\beta + i$ rounds. The same argument holds for the edges incident to $v$.
\end{proof}

%\begin{observation}\label{obs:linial}
%Let $v$ denote an arbitrary vertex in machine $\mu$. If for any vertex or edge $x$ in the $t$-hop of $v$ and an arbitrary number $\alpha$ we have $\hat{s}_{\mu, \alpha}(x) = s_{\alpha}(x)$ then $\hat{s}_{\mu, 2t+\alpha}(v) = s_{2t+\alpha}(v)$ holds. Also for any $e \ni v$, we have $\hat{s}_{\mu, 2t+\alpha-1}(e) = s_{2t+\alpha-1}(e)$. 
%\end{observation}
%\begin{proof}
%	We can simply prove this by induction on $t$. It is easy to check by definition of state-congested local algorithms that this holds for $t=1$. 	
%	
%	For the induction step it suffices to prove the following claim. If for any vertex or edge $x$ in \thop{v}{$\tau$} and an arbitrary number $\tau'$ we have $\hat{s}_{\mu, \tau'}(x) = s_{\tau'}(x)$ then $\hat{s}_{\mu, \tau'+2}(x) = s_{\tau'+2}(x)$ holds for any $x$ in \thop{v}{($\tau$-1)}. This yields from the definition of state-congested local algorithms and the fact that vertices adjacent to any vertex in \thop{v}{($\tau$-1)} are in \thop{v}{$\tau$} as well. Thus $\hat{s}_{\mu, 2\tau+\tau'}(v) = s_{2\tau+\tau'}(v)$ holds for $v$ if for any vertex in \thop{v}{$\tau$} we have $\hat{s}_{\mu, \tau'}(x) = s_{\tau'}(x)$. Also, note that $\hat{s}_{\mu, 2\tau+\tau'-1}(e) = s_{2\tau+\tau'-1}(e)$ holds for any edge $e$ connected to $v$ since the state of any vertex is based on the state of its incident edges in the previous round. 
%\end{proof}

Recall that our goal was to compress $r$ rounds of a low-memory state-congested algorithm $\mathcal{A}$ in few rounds of a low-memory \MPC{} algorithm. By the discussion above, after collecting the $t$-hop of every vertex in $O(\log \log_\Delta n)$ rounds, if $r < t$, then the simulation takes only $O(1)$ extra rounds to complete. For most applications, however, $r$ is much larger than $t$. In such cases, we cannot afford to collect the $r$-hop neighborhood of a vertex in one machine as its size may exceed the space per machine. The idea, here, is to compress every $t$ rounds of $\mathcal{A}$ in $O(1)$ rounds of our low-memory \MPC{} algorithm. To do this, with the above-mentioned approach we can access the state of each vertex and edge after $t$ rounds. The idea, then, is that each machine shares the states of the vertices/edges that it is responsible for with other machines. Once collected this information, each machine then runs $\mathcal{A}$ for another $t$ rounds with respect to the now updated states. By the end of this round, we are aware of the state $s_{2t}(.)$ of each vertex/edge on its responsible machine. We can continue this process for $r/t$ rounds to complete compression of $r$ rounds of $\mathcal{A}$. Overall it takes only $O(\frac{r}{t} + \log t) = O(\frac{r}{\epsilon \log_\Delta n} + \log \log _\Delta n)$ rounds to have the states of all vertices/edges by the end of round $r$.
\end{proof}

\section{Fast \MPC{} Algorithms for Maximal Matching \& MIS}

In this section, we describe our algorithms to find a maximal matching or a maximal independent set of the input graph. We first show in Section~\ref{sec:lowdegree} how we can handle graphs with small maximum degree and then describe how we reduce maximum degree of the input graph in Sections~\ref{sec:partialdegreereduction}, \ref{sec:degreereductioninefficient} using a total memory of $\widetilde{O}(m+n^{1+2\epsilon/3})$ and finally describe the main algorithm in Section~\ref{sec:mainalgorithm} with the optimized total space of size $\widetilde{O}(m)$.

\subsection{Low-Degree Graphs}\label{sec:lowdegree}

In this section, we consider graphs with small maximum degree and show how we can quickly find their maximal matching or MIS.

\begin{lemma}\label{lem:solvinglowdegree}
	For any given graph $G=(V, E)$ of maximum degree $\Delta \leq n^{\epsilon/16}$ where $\epsilon \in (0, 1)$ is a desirably small number that satisfies $\log n = O(n^{\epsilon/4})$, there exists an algorithm that with high probability computes an MIS (or maximal matching) of $G$ in $O(\log \Delta + \log \log n)$ rounds of \MPC{} using $O(n^{\epsilon})$ space per machine and $\widetilde{O}(m)$ total memory.
\end{lemma}

\begin{proof}
	Our first step is to reduce the number of vertices of the graph down to $O(\frac{n}{\Delta^8 \log^2 n})$. To do this, we directly simulate a few rounds of the algorithms of Luby~\cite{DBLP:conf/stoc/Luby85} for MIS and Israeli and Itai~\cite{DBLP:journals/ipl/IsraelI86} for maximal matching. Observe that simulation of each round of these algorithms is straightforward in $O(1)$ rounds of \MPC{} using $\Ot{m}$ total space when the maximum degree is this small and fits the memory of a machine. Both of these algorithms, in each round, reduce the number of edges of the graph by a constant factor in expectation by committing a subset of the vertices to MIS (and removing their neighbors) or by committing a subset of the edges to maximal matching (and removing their incident edges). Therefore, for each of them, it takes only $O(\log(\Delta^9 \log^2 n)) = O(\log \Delta + \log \log n)$ rounds to reduce the number of edges by a factor of $\Delta^9 \log^2 n$ in expectation. Since initially we have at most $n\Delta$ edges in the graph, the remaining graph will have at most $\frac{n\Delta}{\Delta^9 \log^2 n} = \frac{n}{\Delta^8 \log^2 n}$ edges in expectation. Ignoring singleton vertices, the remaining graph cannot have more than $\frac{n}{\Delta^8 \log^2 n}$ vertices in expectation. The success probability can be easily boosted up to high probability by taking $c \log n$ copies of the graph and simulating these algorithms on each instance independently and in parallel and finally choosing the graph whose remaining vertices is the minimum. By a simple application of Chernoff's bound, with probability at least $1-n^{-c}$, the number of vertices is dropped to $\frac{n}{\Delta^8 \log^2 n}$ where we can employ the second part of the algorithm.

	Our second step is to directly simulate $O(\log \Delta)$ rounds of the algorithm of Ghaffari~\cite[Section 3]{DBLP:conf/soda/Ghaffari16} for MIS or Barenboim et al.'s~\cite[Figure~6 -- Phase I]{DBLP:journals/jacm/BarenboimEPS16} for maximal matching.  Both algorithms are also very message efficient and it is also straightforward to simulate each round of them in $O(1)$ rounds of \MPC{} when max degree fits the memory per machine. These algorithms {\em shatter} the graph into smaller connected components of size at most $\Delta^4 \cdot \log n$ by committing a subset of the vertices/edges to MIS/maximal matching. See \cite[Lemma~4.2]{DBLP:conf/soda/Ghaffari16} and \cite[Lemma~4.3]{DBLP:journals/jacm/BarenboimEPS16} for the proof. Since we assume $\Delta \leq n^{\epsilon/10}$ and $\log n = O(n^{\epsilon/4})$, we have $\Delta^4 \cdot \log n \leq n^{\epsilon / 4} \cdot \log n \leq O(n^{\epsilon/2})$ which is substantially smaller than the memory per machine. The diameter of these components also cannot exceed their size. Therefore, we can use Lemma~\ref{lem:collecting} to collect each of these components in a machine in merely $O(\log (\Delta^4\cdot \log n)) = O(\log \Delta + \log \log n)$ rounds. Within a machine, it is then trivial to find MIS/maximal matching  in one round using their corresponding greedy approaches.
	
	It only remains to argue that the total space is only $\Ot{m}$. Recall that Lemma~\ref{lem:collecting}, guarantees that the total space is $\Ot{m' + n'^{1+2\beta}}$ where $n'^{\beta}$ is an upper bound on the size of each component and $n'$ and $m'$ respectively denote the number of vertices and edges of its input graph. Also recall that in the first step of our algorithm, we reduce the vertices by a factor of $\Delta^8\log^2 n$. Therefore we have $n' \leq \frac{n}{\Delta^8 \log^2 n}$. Moreover, the second part of the proof guarantees that the size of no component exceeds $\Delta^4 \log n$. Therefore, overall, collecting the components requires a total space of only
	\begin{equation*}
		\Ot{m' + n'^{1+2\beta}} \leq \widetilde{O}\big(m + n' \cdot (n'^{\beta})^2\big) \leq \widetilde{O}\Big(m + \frac{n}{\Delta^8 \log n} \cdot (\Delta^4 \cdot \log n)^2\Big) \leq \Ot{m + n} \leq \Ot{m}
	\end{equation*}
	as desired. Simulation of the first step requires $O(\log \Delta + \log \log n)$ rounds and simulation of the second step requires $O(\log \Delta)$ rounds. Collecting the components also takes only $O(\log \Delta + \log \log n)$ rounds. Therefore overall the round complexity of the algorithm is $O(\log \Delta + \log \log n)$.
\end{proof}

We note that direct simulation of \local{} algorithms for low-degree graphs would lead to undesirable $n$-dependencies. For instance the algorithm of Ghaffari~\cite{DBLP:conf/soda/Ghaffari16} requires $O(\log \Delta + 2^{\sqrt{\log \log n}})$ rounds for MIS and the algorithm of Barenboim et al.~\cite{DBLP:journals/jacm/BarenboimEPS16} (combined with deterministic maximal matching algorithm of \cite{DBLP:conf/wdag/Fischer17}) requires $O(\log \Delta + \log^3 \log n)$ rounds for maximal matching. We also note that for MIS, a similar approach was used in \cite[Lemma 2.3]{mistreesmpc} to improve over these \local{} bounds for low-degree trees in the \MPC{} model. However, their algorithm requires $O(\log \Delta \cdot \log \log n)$ rounds instead of $O(\log \Delta + \log \log n)$ and uses structural properties of trees to optimize the total memory.

With Lemma~\ref{lem:solvinglowdegree}, if we manage to reduce the maximum degree to $\poly \alpha$, where $\alpha$ is the arboricity of the input graph, then we can solve the problem in $O(\log \alpha + \log \log n)$ rounds. This is the main focus of the forthcoming sections.

\begin{remark}\label{rem:sqrtlogd}
	In an independent paper, Ghaffari and Uitto \cite{sqrtlogd} gave a truly sublinear \MPC{} algorithm for MIS and maximal matching that takes $O(\sqrt{\log \Delta}\cdot \log\log \Delta + \sqrt{\log\log n})$ rounds. While our Lemma~\ref{lem:solvinglowdegree} can be used after degree reduction to imply an $O(\log \alpha + \log^2\log n)$ round algorithm, we can use the result of \cite{sqrtlogd} to slightly improve this bound to $O(\sqrt{\log \alpha}\cdot\log\log \alpha + \log^2\log n)$.
\end{remark}

\subsection{A Partial Degree Reduction Lemma}\label{sec:partialdegreereduction}

Our starting point in this section is (a slightly paraphrased version of) the  degree-reduction algorithm of \cite{DBLP:conf/focs/BarenboimEPS12} for bounded arboricity graphs.

\begin{theorem}[Degree reduction for MIS and maximal matching]\label{thm:barenboim} 
	Let $G=(V, E)$ be a graph with maximum degree $\Delta$ and arboricity $\alpha$ where $\Delta \geq \max\{(5\alpha)^{16}, (5c\log n)^{14}\}$. There exists a low-memory state-congested local algorithm on graph $G$ that takes $\Delta$ as input (i.e., $\Delta$ is initially shared with all nodes) and after $O(\log_{\Delta} n)$ rounds, with probability at least $1-n^{-c}$:
	\begin{enumerate}
		\item Finds an independent set $I \subseteq V$ of $G$ such that each vertex of degree at least $\sqrt{\Delta}$ is either in $I$ or is incident to one vertex in $I$.
		\item Finds a matching $M \subseteq E$ of $G$ such that each vertex of degree at least $\sqrt{\Delta}$ is matched in $M$.
	\end{enumerate}
\end{theorem}

\begin{tboxalg}{Local degree reduction for MIS and maximal matching~\cite{DBLP:conf/focs/BarenboimEPS12}.}\label{alg:degreereduction}
	\begin{enumerate}[label={(\arabic*)}, topsep=0pt,itemsep=0ex,partopsep=0ex,parsep=1ex, leftmargin=*]
			\item Mark every vertex of degree at least $\sqrt{\Delta}$ as ``high-degree" and other vertices as ``low-degree".
			\item Mark a high-degree vertex as ``exposed" if it has at least $\sqrt{\Delta}/2$ low-degree neighbors.
			\item Each exposed vertex picks exactly $\sqrt{\Delta}/2$ of its low-degree neighbors arbitrarily and discards its other edges. Mark a low-degree vertex as ``leaf" if it is now connected to at least one exposed vertex.
			\item Fix $\beta = \Delta^{1/14}$. Mark a leaf vertex $v$ as ``good" if it satisfies the following two conditions: (1) $v$ is connected to less than $\beta$ exposed vertices, (2) $v$ is connected to less than $\beta^2$ other leaves.
			\item \textbf{For MIS:} Draw a random real in $(0, 1)$ for each good leaf. Each good leaf that holds a local minimum number among its good leaf neighbors joins the MIS. We remove the inclusive neighborhood of each vertex that joins the MIS from the graph.\\\vspace{-0.3cm} \\ 
			\textbf{For maximal matching:} Each good leaf $u$ proposes to one of its exposed neighbors uniformly at random. Each exposed vertex receiving at least one proposal accepts one arbitrarily and gets matched to the proposing vertex. We remove the matched vertices from the graph.
	\end{enumerate}
\end{tboxalg}

\begin{lemma}\label{lem:barenboimisincremental}
	Algorithm~\ref{alg:degreereduction} can be completed in $O(1)$ rounds of a low-memory state-congested local algorithm.
\end{lemma}
\begin{proof}
In each step of the algorithm, there are constant possibilities for the state of the edges and the vertices. Also, the random bits just appear in step (5), and each vertex has \O{\log n} random bits. Moreover, it is easy to see that the state of vertices and edges in each step is just based on the states and random bits of their 1-hop. 
 Considering that each step can be simulated in constant rounds of a local algorithm, this algorithm can be completed in $O(1)$ rounds of a state-congested local algorithm. It is also low-memory since state of any vertex in each round can be updated in a space of size \O{\deg(v) · \poly \log n} and for each edge we need space of \O{\poly \log n} bits.
\end{proof}

The following lemma was proved in Theorem 7.2 of \cite{DBLP:conf/focs/BarenboimEPS12}.

\begin{lemma}[\cite{DBLP:conf/focs/BarenboimEPS12}]\label{lem:degreereductionworks}
	Calling Algorithm~\ref{alg:degreereduction} on a graph with maximum degree $\Delta$ and arboricity $\alpha$ where $\Delta \geq \max\{(5\alpha)^{16}, (5c\log n)^{14}\}$ removes $\Delta^{\Omega(1)}$ fraction of its high-degree vertices with probability at least $1-n^{-c}$.
\end{lemma}

Indeed the two lemmas above are sufficient to prove Theorem~\ref{thm:barenboim}.

\begin{proof}[Proof of Theorem~\ref{thm:barenboim}]
	It suffices to iteratively run Algorithm~\ref{alg:degreereduction}. Since each round, by Lemma~\ref{lem:degreereductionworks}, removes $\Delta^{\Omega(1)}$ high-degree vertices, it suffices to run it for only $O(\log_{\Delta}n)$ rounds to remove all high-degree vertices with high probability. Moreover, we showed that each call to Algorithm~\ref{alg:degreereduction} can be completed in $O(1)$ rounds of a state-congested local algorithm; thus, overall, it takes only $O(\log_{\Delta}n)$ rounds of a state-congested local algorithm to remove all high-degree vertices.
\end{proof}

We further show that it takes only $O(1)$ rounds to reduce maximum degree down to $O(n^\epsilon)$ so that it fits the memory per machine.

\begin{lemma} \label{lem:MPC-degree-reduction}
Given a graph $G$ and any desirably small constant $\epsilon \in (0,1)$, 
there exists an algorithm that in \O{1} rounds of an MPC algorithm decreases the maximum degree of the graph to $\O{n^\epsilon}$ using \O{n^\epsilon} space per machine and \Ot{m} total space.
\end{lemma}

\begin{proof}
By Lemma~\ref{lem:degreereductionworks}, we can reduce the maximum degree of the graph to $\O{n^\epsilon}$ by running the Algorithm~\ref{alg:degreereduction} for $O(\log_{n^\epsilon}n) = \O{1}$ time. We just need to show that it is possible to simulate this algorithm is \O{1} rounds of an MPC algorithm using \O{n^\epsilon} space per machine and \Ot{m} total space. Observe that steps (1), (2), (4) and (5) of this algorithm can be simply simulated as separable functions. Therefore, by lemma~\ref{lem:generalmaxload}, it is possible to compute them in \O{1} rounds of MPC using $\O{n^\epsilon}$ per machine and \Ot{m} total space. One can verify that using a similar approach each high-degree vertex can remove all but $\sqrt{\Delta}/2$ of its edges that are connected to low-degree vertices in \O{1} rounds. As a result, the whole algorithm can be simulated in \O{1} rounds of an MPC algorithm using 	$\O{n^\epsilon}$ space per machine and $\Ot{m}$ total space.
\end{proof}

\subsection{Warm-Up: A Simple Algorithm with Inefficient Total Space}\label{sec:degreereductioninefficient}

\begin{theorem}\label{thm:mainbadspace}
	For any given graph $G=(V, E)$ of arboricity $\alpha$, and for any desirably small $\epsilon \in (0,  1)$, there exists an algorithm that with high probability computes a maximal independent set (or maximal matching) of $G$ in $O(\log \alpha + \log^2 \log n)$ rounds of \MPC{} using $O(n^{\epsilon})$ space per machine and $O(m+n^{1+\epsilon/3})$ total memory. The algorithm does not require to know $\alpha$.
\end{theorem}
\begin{proof}

Fix a sufficiently large threshold $\tau = \alpha^{O(1)} + \log^{O(1)} n$. First observe that if $ \Delta \leq \tau$, then we have $\log \Delta = O(\log \alpha + \log \log n)$ and, thus, we can use the algorithm of Lemma~\ref{lem:solvinglowdegree} to solve the problem in $O(\log \Delta + \log \log n) = O(\log \alpha + \log \log n)$ rounds or as described in Remark~\ref{rem:sqrtlogd} to $\widetilde{O}(\sqrt{\log \alpha})$. Therefore, one challenge is to reduce the maximum degree to $\tau$.

The algorithm that we use for this consists of $O(\log \log n)$ {\em phases} (not rounds). Let us denote by $\Delta_i$ the maximum degree of the graph at the start of phase $i$ of the algorithm. The goal is to ensure that in any phase $i$ where $\Delta_i > \tau$, we reduce the maximum degree substantially and get $\Delta_{i+1} \leq \sqrt{\Delta_i}$. Observe, at first, that having this implies that it takes only $O(\log \log n)$ phases to reduce the maximum degree to $\tau$, since otherwise we have
\begin{equation*}
\Delta_{\log \log n} \leq n^{1/2^{\log \log n}} = n^{1/\log n} \leq O(1).
\end{equation*}
To achieve the goal of reducing maximum degree from $\Delta$ to $\sqrt{\Delta}$ by the end of each phase, we employ Algorithm~\ref{alg:degreereduction} which precisely guarantees this by Theorem~\ref{thm:barenboim}. Note that our algorithm is not given the arboricity $\alpha$ of the graph and, thus, we do not know the value of $\tau$ and cannot check whether $\Delta < \tau$. However, if Algorithm~\ref{alg:degreereduction} fails, which we are able to check by computing the maximum degree of the remaining graph, we can be sure that $\Delta < \tau$. Unfortunately, direct simulation of Theorem~\ref{thm:barenboim} is infeasible as it takes up to $O(\log_\Delta n)$ rounds which gets close to $O(\log n)$ as $\Delta$ gets smaller and smaller. To resolve this, we use the blind coordination lemma (Lemma~\ref{lem:blindcoordination}) to compress multiple rounds of Algorithm~\ref{alg:degreereduction} in a few rounds of MPC.
Recall that by Lemma~\ref{lem:blindcoordination}, if $\Delta \leq n^\epsilon$, it takes only $O(\frac{r}{\log_\Delta n} + \log \log_\Delta n)$ rounds to run $r$ rounds of any low-memory state-congested local algorithm with a low-memory \MPC{} algorithm and by Lemma~\ref{lem:MPC-degree-reduction}, is possible to decrease the maximum degree to $n^\epsilon$ in constant rounds of MPC. Therefore, after decreasing $\Delta$ to $n^\epsilon$, we run $O(\log_\Delta n)$ rounds of Algorithm~\ref{alg:degreereduction}, which we proved is a low-memory state-congested algorithm in Lemma~\ref{lem:barenboimisincremental}. Note that, this process takes only $O(\log \log_\Delta n)$ rounds of MPC. Overall, since we have $O(\log \log n)$ phases each taking $O(\log \log_\Delta n)$ rounds, it takes only $O(\log^2 \log n)$ rounds to reduce maximum degree to $\tau$. Therefore, the final running time of the algorithm is $O(\sqrt{\log \alpha} \cdot \log\log \alpha + \log^2 \log n)$.\end{proof}

\subsection{The Main Algorithm}\label{sec:mainalgorithm}

Observe that although the local space of each machine in Theorem~\ref{thm:mainbadspace} is only $O(n^\epsilon)$, the aggregated space over all machines is $n^{1+\Omega(\epsilon)}$ which may be much larger than the optimal total space of $\Ot{m}$ that suffices to store the original input. In this section, we resolve this shortcoming by modifying our algorithm to achieve an optimal total space of $\Ot{m}$. We note that these modifications, remarkably, do not lead to any blow-up in the round complexity of the algorithm. 

\restate{Theorem~\ref{thm:main}}{
	\mainthm{}
}

The main reason that our algorithm for Theorem~\ref{thm:mainbadspace} requires a total space of at least $n^{1+\Omega(\epsilon)}$ is the blind coordination lemma. This blow-up in total space comes from the fact that for {\em each} vertex, we collect its neighborhood of size up to $n^{\Omega(\epsilon)}$ in the machine that is responsible for it. Therefore, inevitably, we need a total space of $n^{1+\Omega(\epsilon)}$ to store these neighborhoods for all the vertices. To alleviate this, we exploit several structural properties of Algorithm~\ref{alg:degreereduction} to  employ the blind coordination procedure on only a carefully picked subset of the vertices that we dynamically update over different rounds/phases of the algorithm.

Recall that our algorithm for Theorem~\ref{thm:mainbadspace} is composed of $O(\log \log n)$ phases that in turn reduce the maximum degree from $\Delta$ to $\sqrt{\Delta}$ (where $\Delta$ is the maximum degree in the remaining graph by the end of the previous phase). Indeed the only part of the algorithm that requires the blind coordination lemma and, thus, a total space of $n^{1+\Omega(\epsilon)}$ is completing each of these phases. Therefore, to reduce the total space to $O(m)$, it suffices to prove the following lemma.

\begin{lemma}
	For any graph $G$ with maximum degree $\Delta$ and arboricity $\alpha$, and for any desirably small constant $\epsilon \in  (0, 1)$, there exists an algorithm that finds with high probability, an independent set $I$ (resp. a matching $M$) of $G$ in $O(\log \log n)$ rounds of \MPC{} with $O(n^{\epsilon})$ space per machine and $\widetilde{O}(m)$ total memory, such that the maximum degree of $G[V\setminus I]$ (resp. $G[V\setminus V(M)]$) is at most $\max\{\sqrt{\Delta}, \alpha^{O(1)} + \log^{O(1)} n\}$.
\end{lemma}
\begin{proof}

Define $B := \{v \in V \, | \, \deg(v) < \sqrt{\Delta}, \, \max_{u \in N(v)} \deg(u) < \sqrt{\Delta} \}$ to be the set of low-degree vertices at the start of Algorithm~\ref{alg:degreereduction} that have no high-degree neighbors. Note that we simply assume that $\Delta = \O{n^\epsilon}$ since by Lemma~\ref{lem:MPC-degree-reduction} it is possible to reduce the maximum degree of the graph to $\O{n^\epsilon}$ in \O{1} rounds on MPC. We first note the following property of Algorithm~\ref{alg:degreereduction}.

\begin{observation}\label{obs:lowdegb}
	The output of Algorithm~\ref{alg:degreereduction} on graph $G$ is the same as its output on graph $G[V \setminus B]$.
\end{observation}

Algorithm~\ref{alg:degreereduction} marks each vertex of degree at least $\sqrt{\Delta}$ as high-degree and with $O(\log_\Delta n)$ calls to it, each of the high-degree vertices either gets removed from the graph or its degree drops to less than $\sqrt{\Delta}$. Throughout this process, after each call to Algorithm~\ref{alg:degreereduction}, a subset of high-degree vertices gets removed from the graph. Therefore, a low-degree vertex $v$ that is initially incident to a high-degree vertex $u$ and is thus not in $B$, may join $B$ after $u$ is removed. Upon joining $B$, we change the state of the vertex to ``dead". By Observation~\ref{obs:lowdegb}, a vertex that is marked as dead will have no impact on the outcome of the rest of the algorithm until the degree of every vertex drops down to $\sqrt{\Delta}$. Using this, we first show that it is possible to compress multiple rounds by collecting the neighborhood of only the high-degree vertices. Then, we explain how we manage to store the neighborhood of all the high degree vertices with using $\O{m}$ total memory. The overall idea is that after removing a portion of the high-degree vertices in each phase of the algorithm that, we use this extra space to expand the neighborhood of the remaining ones. Finally ,we show how we actually gather the neighborhood of the vertices by adapting the exponential growth technique. 

\vspace{-0.3cm}
\paragraph{Round compression without collecting the neighborhood of low-degree vertices.}

 Instead of initially collecting the $\Omega(\log_\Delta n)$-hop of {\em every} vertex to compress the phases, which is inefficient in terms of the total space used, we collect the neighborhoods of only the high-degree vertices. Interestingly, if we just collect the $t$-hop of the high-degree vertices, we end up having the $(t-1)$-hop of all the vertices that we care about in at least one machine. By Observation~\ref{obs:lowdegb}, any low-degree vertex that has an {\em impact} on the outcome of the algorithm, is connected to at least one high-degree vertex. Take a low-degree vertex $v$ that is incident to a high-degree $u$. Since we collect the $t$-hop of $u$ in the machine that is responsible for $u$, we also have access to the $(t-1)$-hop of $v$ in that machine. Note that we might have $(t-1)$-hop of some low-degree vertices in more than one machine but it does not cause a problem for us. Now, suppose that we draw for any vertex $\O{\log n}$ random real number in (0, 1), then collect the $t$-hop of every high-degree vertex along their random bits in a machine responsible for it. As a result of this, for any vertex, we have its $(t-1)$-hop in at least one machine. This allows us to compute the correct state of all the vertices after  $r = \Theta(t-1)$ calls to Algorithm~\ref{alg:degreereduction}. The reason is that Algorithm~\ref{alg:degreereduction} uses at most one random number per vertex and by Lemma~\ref{lem:barenboimisincremental}, $r$ calls to this algorithm can be simulated as $\O{r}$ rounds of a low-memory state-congested local algorithm. 
  Also, by Observation~\ref{obs:linial} the state of any vertex after these $\O{r}$ rounds can be computed in one round of \MPC{} given the initial state of its $t'$-hop where $t'=\Theta(r)$.  Note that in this case, the initial state of vertices is their random numbers. Therefore,  having the $t$-hop of high-degree vertices suffices to have the state of all the vertices after $\Theta(t)$ calls to Algorithm~\ref{alg:degreereduction}.

\vspace{-0.3cm}
\paragraph{Handling the high-degree vertices.} 
We showed that  it suffices to collect the neighborhood around only the high-degree vertices to be able to compress multiple rounds of the algorithm. 
However, even storing a neighborhood of size up to $n^{O(\epsilon)}$ for high-degree vertices may require much more than $O(m)$ overall space. To resolve this, we set a capacity $s$ on the size of the neighborhood that we collect for each of the high-degree vertices and update this capacity iteratively. The initial capacity on each of the high-degree vertices is $s_0 = O(1)$. This means that it is initially impossible to compress multiple rounds of the algorithm.  Recall that by Lemma~\ref{lem:degreereductionworks}, each time we call Algorithm~\ref{alg:degreereduction}, with high probability at least $\Delta^{\delta}$ fraction of the high-degree vertices will be removed from the graph for some constant $\delta > 0$. Thus, after $1/\delta$ calls to Algorithm~\ref{alg:degreereduction}, at least $\Delta$ fraction of the high-degree vertices are removed and each high-degree vertex affords to collect its direct neighbors in its machine. We then continue simulating the algorithm for $2/\delta = O(1)$ more rounds in $d = O(1)$ rounds of our \MPC{} algorithm without any compression. This allows us to increase the capacity of the remaining vertices by a factor of $\Delta^2$ while keeping the total capacity of the remaining vertices the same. That is, we have $s_1 = \Delta \cdot \Delta^2$ and now have enough space to collect the 3-hop of all the high-degree vertices in one machine. As a result of this we have 2-hop of any vertex in at least one machine and we can now simulate two times more number of rounds of the algorithm in $d$ rounds of MPC. This reduces the number of remaining high-degree vertices by a factor of $\Delta^4$ and we can increase the capacity by this factor, achieving $s_2 = \Delta^3 \cdot \Delta^4 = \Delta^7$. Now, after collecting the 7-hop of every high-degree vertex, we can run 6 times more number of rounds of the algorithm in $d$ rounds and reduce the high-degree vertices by a factor of $\Delta^{13}$ and get $s_3 = \Delta^7 \cdot \Delta^{13}$. Overall, the capacity is increased double-exponentially in each step and it takes only $O(\log \log_\Delta n)$ rounds to get a capacity that is essentially as large as the space $O(n^{\epsilon})$ of a machine. 

\vspace{-0.3cm}
\paragraph{Adapted exponential growth technique.}
The only missing part of the proof is about how we gather the $t$-hop of high-degree vertices in a single machine and how we expand it in each iteration of the algorithm. Needless to mention that we are not concerned about the low-degree vertices that are not connected to any high-degree vertex; therefore, after any iteration of the algorithm we remove all such vertices from all the machines. Note that simply using Lemma~\ref{lem:collecting} in each iteration is not efficient since it collects the neighborhood of all the vertices. It also takes $\log{ \log (t)}$ rounds to gather the $t$-hop of even a single vertex in one machine but we expect each iteration of our algorithm to take \O{1} rounds. To overcome these issues we take the exponential growth algorithm of Lemma~\ref{lem:collecting} and adapt it to our needs. Roughly speaking, we claim that if we just call this algorithm on high degree vertices we can get $t$-hop of them in $\log \log(t)$ rounds. Assume that the $t$-hop of all high-degree vertices is gathered in the machine responsible for them. For any high-degree vertex $v$ we just send requests to gather the neighborhood of the high-degree vertices in $t$-hop of $v$ in the machine responsible for that. This gives us the $(2t-2)$-hop of vertex $v$ assuming that it does not violate the memory limits. The reason is that for any low-degree vertex $u$ in $t$-hop of $v$ there is also at least a high-degree vertex $u'$ in $t$-hop of $v$ where the distance between $u$ and $u'$ is at most two. Without loss of generality, we assume that $t > 3$ otherwise we simply gather the $t$-hop in $t$ rounds. Observe that if for any high-degree vertex we set its $3$-hop as its initial neighborhood, using the mentioned algorithm in $\log\log(t)$ rounds we gather the $t$-hop of high-degree vertices in the machine responsible for them. Also to decrease the number of rounds in each iteration we simply start gathering the neighborhood of vertices from where we left off in the previous iteration. Assume that we had the $t'$-hop of all the high-degree vertices in the previous integration. After that iteration is completed, we remove the neighborhood of all the vertices that are no longer high-degree and for those that are still-high degree we send requests to the high-degree vertices in their neighborhood. As a result we get the $(2t'-2)$-hop of all the remaining high-degree vertices. Note that in any iteration this algorithm does not need space more than $\O{m}$ since the size of the neighborhood of each high-degree vertex that we gather in its machine increases double exponentially and so does its budget. However, it is possible that the extra space that we get by removing the data that is no-loner needed is not in the machines that we need. We can simply manage it by redistributing our data in the machines using a deterministic hash function. We just need the machine responsible for any vertex to know its new location and its state (whether it is dead, low-degree, or high-degree).

To sum up, we prove this lemma by simulating $O(\log_\Delta n)$ rounds of Algorithm~\ref{alg:degreereduction} in just  $O(\log \log_\Delta n)$ rounds of an MPC algorithm using $\O{n^\epsilon}$ space per machine and \Ot{m} overall space. The main difficulty that we face is minimizing the space that our algorithm uses. We handle that by gathering the neighborhood of just the high-degree vertices and controlling the space used by any high-degree vertex by a dynamic budget.  We  also provide a technique to update the neighborhood of these vertices throughout the algorithm.
\end{proof}

\section{Load Balancing}\label{sec:load-balancing}

This section addresses one of the technical details faced by low-memory \MPC{} graph algorithms. Assume that we are given a function which we need to compute for all the vertices in the graph. For any vertex $v$ the value of this function is based on its neighbors. The complication here arises from the fact that the degree of the vertices can be larger than the memory of the machines. Therefore, we are not able to simply gather neighbors of $v$ in one machine and compute the given function. Examples of such functions that we need in our algorithms are as follows.
\begin{itemize}
\item Finding the degree of the vertices.
\item For any vertex in the graph find the neighbor with the minimum label. One usage of this is in the Luby's algorithm when each vertex picks a random number and we need to find the vertices who have the minimum number among their neighbors.
\item Some vertices of the graph are chosen to be in the MIS and each vertex needs to know whether it is adjacent to any such vertex or not.
\end{itemize}

To give a general algorithm that applies to all such function we defined \emph{separable functions} in Definition \ref{def:separable}. One can easily see that all the mentioned problems can be modeled by such a function. In Lemma~\ref{lem:generalmaxload} we prove that there is an algorithm that solves these problems in $\O{1/\epsilon}$ rounds of MPC using $O(n^{\epsilon})$ space per machine and $O(m)$ total space. We first need the following auxiliary lemma.

\begin{lemma}\label{lem:maxload}
Let $A$ be a set of real numbers and let $\sigma$ be a constant number in $(0,1)$ where for any $x\in A$ we have $1<x<n^{\sigma}$.
There exists a hash function that uses $\O{\log p}$ random bits and for any integer $p$ distributes elements of $A$ into $p$ partitions such that, with high probability, sum of the numbers in each partition is $\O{(\sum_{x\in A}{x}/p + n^{\sigma})\log{p}}$.
\end{lemma}

\begin{proof}
We first partition elements of $A$, based on their weights, into $p$ subsets of size at most $p$  which we denote by $A_1, \cdots A_p$. Subset $A_1$ contains the $ m :=\lceil |A|/p \rceil$ greatest numbers and $A_p$ contains the smallest numbers. Observe that if we pick an arbitrary number from each subset, their summation is $\O{\sum_{x\in A}{x}/p + n^{\sigma}}$. Also, as an application of balls and bins problem which is formally proven in \cite{DBLP:conf/soda/SchmidtSS93}, there exists a hash function that using $\O{\log{p}}$ random bits, distributes $p$ balls into $p$ bins such that, with high probability, the maximum load of the bins is $\O{\log{p}}$. Therefore, if we distribute elements of $A$ into $p$ partitions using such a hash function, w.h.p., the summation of the numbers in each partition is $\O{(\sum_{x\in A}{x}/p + n^{\sigma})\log{p}}$.
\end{proof}

\restate{Lemma~\ref{lem:generalmaxload}}{\separablelemma{}}

\begin{proof}

Assume that we have $\O{(m/n^\epsilon) \log{n}}$ machines, each with space $\O{n^\epsilon}$. Draw $\log n$ random bits and share them with all the machines. By Lemma~\ref{lem:maxload}, there exists a hash function $h$ that using these random bits gives an assignment of vertices to machines such that w.h.p., the overall degree of the vertices assigned to each machine is bounded by $\O{n^\epsilon + \Delta + \log}$ where $\Delta$ is the maximum degree in the graph.  If $\Delta = \O{n^{\epsilon}}$, for any vertex $v$ we can gather $n_v$ in the machine that $v$ is assigned to and compute $F_v(n_v)$. However, this is not a valid assumption since $\Delta$ can be as large as $n$. In that case, we use the fact that $F_v$ is a separable function. Let $d_v$ denote the number of machines that contain a piece of information that we need to compute function $F_v$. At the beginning, $d_v$ is bounded by the number of machines.  Since $F$ is a separable function, to be able to complete the proof using Lemma~\ref{lem:maxload}, it suffices if for any vertex $v$ we somehow limit $d_v$ by $\O{n^{\epsilon}}$. For instance, assume that function $F_v$ is the degree of vertex $v$. After distributing the edges, in any machine  $m$ we compute the number of edges that vertex $v$ has in $m$; therefore, the size of the data needed to compute $F_v$ decreases to the number of machines that have at least one edge of $v$. Denote by $\Delta'$ the maximum of $d_v$ among all the vertices. We give an algorithm that in each iteration decrease $\Delta'$ by a factor of $n^\epsilon$. We first cluster the machines into bundles of size $\Delta'/n^{\epsilon}$. The total space of each bundle is $\Delta'$. We treat each bundle as a machine with space $\Delta'$, and use the hash function $h$ to give an assignment of vertices to the machines (which are bundles here). Consider a piece of information related to a vertex that is assigned to bundle $b$. This algorithm randomly sends it to one of the machines in this bundle. In this way, w.h.p., none of the machines is overloaded and $\Delta'$ decreases by a factor of $n^\epsilon$. After repeating this for $\O{1/\epsilon}$ times, w.h.p., $\Delta'$ decreases to one and we simply compute $F_v$ for any vertex $v$ in a single machine.   
\end{proof}

\section{Acknowledgements}
The authors would like to thank Saeed Seddighin for useful discussions and for bringing up the question of whether matching can be solved in sublogarithmic rounds of low-memory \MPC{}.

\bibliographystyle{plain}
\bibliography{references}
	
\end{document}